\documentclass[12pt]{article}
\usepackage[utf8]{inputenc}
\usepackage[english]{babel}
\usepackage{hyperref}
\usepackage{amsthm,thmtools}
\usepackage{amsfonts}
\usepackage{amssymb}
\usepackage{mathtools}
\usepackage{amsmath}
\usepackage{comment}

\newtheorem{proposition}{Proposition}[section]
\newtheorem{corollary}{Corollary}[section]
\newtheorem{theorem}{Theorem}[section]
\newtheorem{lemma}{Lemma}[section]
\newtheorem{remark}[theorem]{Remark}

\usepackage{cite}
\usepackage{subcaption}
\usepackage{graphicx}
\usepackage{geometry}
\usepackage{subfiles}
\usepackage{mathtools}
\usepackage{pgfplots}
\pgfplotsset{/pgf/number format/use comma,compat=newest}
\usepackage{url}
\usepackage{tikz}
\usetikzlibrary{arrows,decorations.pathmorphing,backgrounds,positioning,fit,petri,patterns}

\newcommand{\GO}[1]{{\color{red}{GO:} #1}}

\title{Limit theorems for the cubic mean-field Ising model}
\usepackage{authblk}
\author{Pierluigi Contucci, Emanuele Mingione, Godwin Osabutey}
\affil{Department of Mathematics, University of Bologna}

\begin{document}
\date{\today}

\maketitle

\begin{abstract}
We study a mean-field spin model with three- and two-body interactions. The equilibrium measure for large volumes is shown to have three pure states, the phases of the model. They include the two with opposite magnetization and an unpolarized one with zero magnetization, merging at the critical point. We prove that the central limit theorem holds for a suitably rescaled magnetization, while its violation with the typical quartic behavior appears at the critical point. 
\end{abstract}

\section{Introduction}

In this paper, we investigate the mean-field Ising spin model with quadratic and cubic interactions. The interest in such a model comes from two large fields of research. The first is condensed matter physics, where the three-body interaction plays a role in the description of the phase separation phenomena of some magnetic alloys \cite{SubLeb1999} lacking spin-flip symmetry. Those physical systems cannot be described by the sole use of a two-body interaction, while a three-body term captures some features of their behavior \cite{Kadanoff1971}. This fact is well paralleled by the Ginibre theorem about functions of spin configurations that are fully classified by an orthonormal base of $k$-body interactions \cite{Ginibre1970}. Those physical phenomena are well described by statistical mechanics models on regular lattices in finite ($d$=2,3) dimensions. While some of those models have an exact solution in very special cases \cite{BaxterWu1973,BaxterWu1974}, it is well known that the mean-field approximation provides an analytically viable setting and a fair description of the phase separation. In those cases, the term mean-field approximation is understood in the sense of a special class of probability measure where the Boltzmann-Gibbs variational principle is optimised: instead of minimizing the free energy over all probability measures, one restricts it to product measures on single spins \cite{FroyenSH1976, Bidaux1986}. 

The other field in which the three-body interactions came to play a role is that of the applications to complex systems, in particular those of socio-technical nature \cite{ContucciKO2022} where the social network structure with long-range interaction represents a realistic description of the phenomenon and not an approximation of its finite-dimensional version \cite{ACMM2017, Battiston2020Beyond_pairwise,Bianconi2021Book, Benson2018}. 
In this case, from a mathematical perspective, the introduction of the three-body interaction entails moving from a  graph-theoretical environment of vertices and edges to a richer hypergraph setting where the three-body terms, representing the faces of the hypergraph, are also taken into account. 

The presence of the cubic interactions brings technical difficulties in the analysis of the model. In particular, the non-convex energy contribution due to the cubic power prevents the use of the Hubbard-Stratonovich transform, which instead is very efficient in the  case of quadratic interactions. More precisely, even if the thermodynamic limit of the free energy can be easily computed by large deviation arguments, the fluctuations of the order parameter cannot be analysed with the classical rigorous methods for a mean-field system with pairwise interaction \cite{Ellis85,EllisNewman78,EllisNR80}.
In order to overcome this obstacle we need a fine control on the $N$-asymptotic behavior of the partition function that is obtained by a method similar to that recently introduced in \cite{Mukherjee2021}.  

This paper presents a rigorous analysis of the mean-field model with three- and two-body interactions in a zero magnetic field. We show that the infinite-volume properties of the model display new phenomena that are absent in the quadratic mean-field case. In particular, we prove that the equilibria of the system include not only  positively and negatively polarized states but also an unpolarized stable state in the presence of a non-zero cubic term that breaks the spin-flip symmetry. Finally, we also study the fluctuation of the magnetization in the entire phase space, specifying the behavior at phase separation and at the critical point. The critical exponent for the magnetization, moreover, takes on a value of zero towards the unpolarized directions of the phase space, and phase transitions can occur in the antiferromagnetic region. 
 
This paper is organised as follows: Section \ref{sec2:def and results} contains the formal definition of the model as well as a statement of the main results. In Section \ref{sec3:proofs}, we study the properties of the consistency equation that describes the system in its  stationary equilibrium state. These properties provide an analytical description of the system's phase diagram and the magnetization's limiting behavior, as well as the computation of the critical exponents. Finally, Section \ref{sec4:conclusion} contains conclusions and perspectives, and the Appendices \ref{appendix:A} and \ref{appendix:B} contain technical and concentration results used throughout the work.

\section{Definitions and main results}\label{sec2:def and results}
Let us consider $N$ spins $\sigma=(\sigma_i)_{i\leq N}\in \{ -1, +1\}^N$ interacting  through an Hamiltonian of the form
\begin{equation}\label{Cubic_model}
    H_N(\sigma) = -\dfrac{K}{3N^2} \sum_{i,j,k = 1}^N  \sigma_i\sigma_j\sigma_k - \dfrac{J}{2N}\sum_{i,j = 1}^N \sigma_i\sigma_j - h\sum_{i = 1}^N \sigma_i \; ,
\end{equation}
where $(K,J,h)\in\mathbb{R}^3$, $K$ and $J$ tune the  interactions among triples and pairs of spins, respectively, while $h$ represents an external field acting on the system. When $K=0$, the previous Hamiltonian reduces to the well-know Curie-Weiss case. In this work we will concentrate on the case $h=0$ and use the parameter $K$ as a spin-flip symmetry breaking term reducing \eqref{Cubic_model} to an Hamiltonian that can be represented as  
\begin{equation}\label{KJ model}
H_N(\sigma) = - N\left(\frac{K}{3}\, m_N^3(\sigma)+ \frac{J}{2}\, m_N^2(\sigma)\right)
\end{equation} 
where $m_N$ is the magnetization per particle:
\begin{equation}\label{magetisation}
    m_N(\sigma) = \frac{1}{N} \sum_{i = 1}^{N} \sigma_i.
\end{equation}

The expression \eqref{KJ model} highlights the mean-field nature of the model. The  Boltzmann-Gibbs  probability measure associated to $H_N$  is 
\begin{equation}\label{Gibbs}
    \mu_{N} (\sigma) = \frac{e^{-  H_N(\sigma)}}{Z_N},
\end{equation}
where $Z_N = \sum_{\sigma\in\{-1,+1\}^N}\exp\left(-H_N(\sigma)\right)$ is the partition function. In equation \eqref{Gibbs}, we set the usual inverse temperature $\beta$ to $1$ without loss since it has been reabsorbed in the parameters of the model. Notice that since the Hamiltonian \eqref{KJ model} is invariant 
under the transformation $K \mapsto -K$, and $\sigma_i \mapsto -\sigma_i$ for $i=1,...,N$, one can study the model only for $K>0$ without loss.

Our aim is to obtain a complete characterization of the model's \textit{phase diagram}, an analysis of the asymptotic distribution of the magnetization in the presence and absence of phase transitions, the \textit{fluctuations} of the suitably rescaled magnetization \eqref{magetisation} w.r.t. the Boltzmann-Gibbs measure \eqref{Gibbs} at and away from the critical point, and the computation of the \textit{critical exponents}. 

All the above  properties are strictly related to the analytical properties of the \textit{free energy} of the system, which is the starting point of our analysis. Let us define  the thermodynamic pressure, i.e., the generating functional as:

\begin{equation}\label{pressurepp}
p_N = \dfrac{1}{N}\,\log Z_N  \; .
\end{equation}
Notice that $p_N$ equals the free energy up to a minus sign. 
The thermodynamic limit of \eqref{pressurepp} can be easily computed applying  Varadhan's  integral lemma \cite{Ellis85, DZ98}, obtaining:

\begin{proposition} Given $(K,J) \in \mathbb{R}^2$ the limiting pressure of \eqref{pressurepp} admits the following variational representation:
\begin{equation}\label{varpressure}
p := \lim_{N\to\infty} p_N = \sup_{m\in[-1,1]} \phi(m),
\end{equation}
where $\phi(m)= u(m)- I(m)$ with
\begin{equation}\label{energy}
u(m)= \dfrac{K}{3} m^3 +\dfrac{J}{2} m^2 
\end{equation}
is the energy contribution and
\begin{equation}\label{entropy}
I(m)=\frac{1-m}{2}\log\left(\frac{1-m}{2}\right)+\frac{1+m}{2}\log\left(\frac{1+m}{2}\right)
\end{equation}
is the binary entropy contribution. 
\end{proposition}
The critical points of \eqref{varpressure}  satisfy the consistency equation, 
\begin{equation}\label{mean-field eqn}
    m = \tanh(Km^2+Jm).
\end{equation}
A careful analysis shows that, among the solutions of \eqref{mean-field eqn},  the function $\phi(m)$ in \eqref{varpressure} can have one or two global maximizers in the interval $(-1,1)$ for fixed $(K,J)$ (see Figure \ref{KJ}). 
\begin{figure}[ht!]
\includegraphics[width=14cm, height=8.3cm]{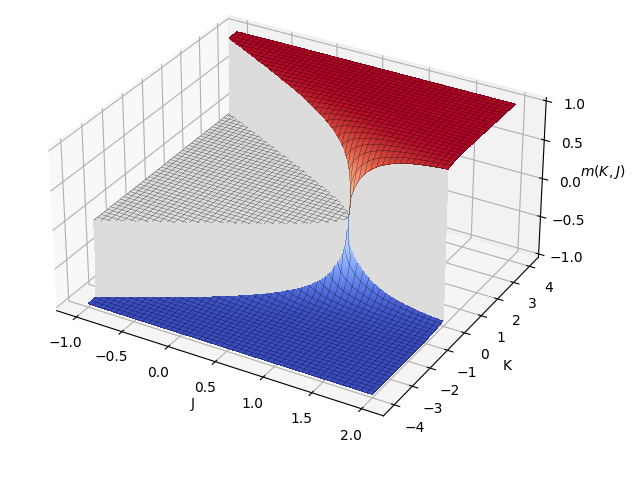}
\caption{Stable solutions of the mean-field equation as a function of $K$ and $J$. There are three stable phases presented here: the positive polarized phase depicted in red, the unpolarized phase given as the gray plateau, and the negative polarized phase denoted by the blue color. At the critical point, $(K,J)=(0,1)$, the three phases of the cubic model as well as the two phases of the Curie-Weiss plane ($K=0$) coalesce.}\label{KJ}
\end{figure}

In particular, we can divide the parameter space $(K,J)\in \mathbb{R}_+\times \mathbb{R}$  accordingly to the following:

\begin{proposition}[Phase diagram]\label{prop uniq}
For any $K>0$, there exists $J=\gamma(K)$ defined in Proposition \ref{wall} such that the function  $m \mapsto \phi(m)$ has a unique maximum point $m^*$ for $(K,J)\in (\mathbb{R}_+\times\mathbb{R}) \backslash \gamma$. Moreover, on the curve $\gamma$ there are two global maximizers, $0=m_0<m_1$ and the limit as $K\rightarrow0$ of $\gamma(K)$ identifies the critical point $(K_c,J_c)=(0,1)$ where the magnetization takes the value $m_c=0$.
\end{proposition}

In physical terms, the presence of two global maximizers  corresponds to the existence of two  different thermodynamic equilibrium phases, whereas the curve $\gamma$ represents the  coexistence curve. Let's note that $m_0$ and $m_1$ represent a stable paramagnetic state and a positively polarised state, respectively. The paramagnetic state is characterized by the absence of spontaneous magnetic order and the presence of symmetry between the up and down spin, with no preference for either direction. The jump from the paramagnetic state to the polarized state, namely when the magnetization jumps from $m_0$ to $m_1$, represents a \textit{first-order phase transition}, which is markedly different from the quadratic mean-field model ($K=0$) having a second-order phase transition in $J$. More precisely if we denotes by $m^*(K,J)$ the unique maximizer of $\phi$, for any $\bar{K}>0$ there exists   $\bar{J}=\gamma(\bar{K})\in (-\infty,1)$ such that 
$$
0=\lim_{J\to \bar{J}^-} m(J,\bar{K})\neq  \lim_{J\to \bar{J}^+} m(J,\bar{K})>0.
$$
This behavior  is somehow reminiscent of the Curie-Weiss Potts model analyzed in \cite{EllisWang1990} where for any value of the parameter $q$ a first order phase transition is observed.
Numerical simulations  of the  phase diagram described in Proposition \ref{prop uniq} can be seen in Figure \ref{fig:maxi}.

\begin{figure}[ht!]
\includegraphics[width=13cm]{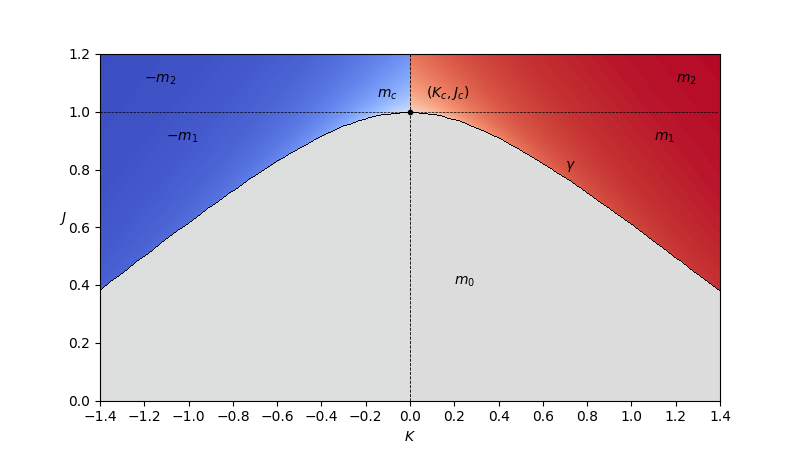}
    \caption{Phase diagram of the model with coexistence curve $\gamma$ and the critical point $(K_c,J_c)$ in the ($K,J$) plane.}
    \label{fig:maxi}
\end{figure}

In the standard Curie-Weiss model, when $J>0$ we know that as soon as $h>0$ one obtains a positive magnetization. The reason is that the energy contribution due to $h$ favors only spins aligned with $sign(h)$. On the contrary, in our system, $J,K>0$, the energy contribution  due to $K$ can be minimized by  configurations containing  both up and down spin signs. This implies that the  entropy contribution can dominate also for small but non-zero $K$, giving a zero magnetization.

The  next theorem contains the law of large numbers and the central limit theorem for the distribution of $m_N(\sigma)$ with respect to the Boltzmann-Gibbs measure. 

\begin{theorem}[Asymptotic distribution of the magnetization]\label{asymptotic_mag}  
Consider the 
Hamiltonian in \eqref{KJ model}, then the following holds:
\begin{enumerate}
    \item For $ (K,J)\in (\mathbb{R}_+\times\mathbb{R})   \backslash (\gamma\cup (K_c,J_c))$ the function  $\phi(m)$ in \eqref{varpressure} has a unique global maximizer $m^*$ such that $\phi''(m^*)<0$ and 

     \begin{equation}\label{lln_m*}
        m_N \xrightarrow[N\rightarrow\infty]{\mathcal{D}} \delta_{m^*}.
    \end{equation}
     Moreover,
     \begin{equation}\label{asymptotic_uniq}
        N^{\frac{1}{2}}(m_N-m^*) \xrightarrow[N\rightarrow\infty]{\mathcal{D}} \mathcal{N}\bigg(0,-\dfrac{1}{\phi''(m^*)}\bigg).
    \end{equation}

    \item Given  $(K,J)\in\gamma$ we  denote by $m_0<m_1$ the two global  maximizers of $\phi(m)$. For $i\in\{0,1\}$ we define the quantity
    \begin{equation}\label{rho_i mi}
        \rho_i := \frac{[(m_i^2-1)\phi''(m_i)]^{-\frac{1}{2}}}{[(m_0^2-1)\phi''(m_0)]^{-\frac{1}{2}}+[(m_1^2-1)\phi''(m_1)]^{-\frac{1}{2}}}.
    \end{equation}
    Then we have that \begin{equation}\label{lln_mi}
        m_N \xrightarrow[N\rightarrow\infty]{\mathcal{D}} \sum_{i\in\{0,1\}}\rho_i\delta_{m_i}.
    \end{equation}
    
    Moreover let   $A_i\subseteq[-1,1]$ be an interval containing $m_i$ in its interior such that $\phi(m_i)>\phi(m)$ for all $m\in cl(A_i)\backslash\{m_i\}$, then

    \begin{equation}\label{asymptotic_mi}
        N^{\frac{1}{2}}(m_N-m_i)\big|\{m_N\in A_i\} \xrightarrow[N\rightarrow\infty]{\mathcal{D}} \mathcal{N}\bigg(0,-\dfrac{1}{\phi''(m_i)}\bigg).
    \end{equation}

    \item At the critical point ($K_c,J_c$), we have that
     \begin{equation}\label{lln_mc}
        m_N \xrightarrow[N\rightarrow\infty]{\mathcal{D}} \delta_{0}.
    \end{equation}
    Moreover, 
    \begin{equation}\label{asymptotic_cpt}
        N^{\frac{1}{4}}\,m_N \xrightarrow[N\rightarrow\infty]{\mathcal{D}} C \exp{\bigg(\dfrac{\phi^{(4)}(0)}{24} x^4\bigg)} dx = C \exp{\bigg(\frac{-x^4}{12} \bigg)}dx,
    \end{equation}
    where $\phi^{(4)}(0) =-2$ denote the fourth derivative of $\phi(m)$ evaluated at $m=0$ and
    
    $C^{-1}={\displaystyle \int_{-\infty}^\infty} \exp{\bigg(\dfrac{-x^4}{12} \bigg)}dx = \frac{\sqrt[4]{3} \; \Gamma(\frac{1}{4})}{\sqrt{2}}$.
\end{enumerate}
\end{theorem}

Finally, we study the behavior of the limiting value of the magnetization  near the critical point $(K_c,J_c)=(0,1)$ namely the critical exponents of the model. The average value of the magnetization is  given by the LLN in Theorem \ref{asymptotic_mag} and will be denoted by $m^*(K,J)$.
The following proposition describes the critical behavior of  $m^*(K,J)$  when  $(K,J)\to (K_c,J_c)$ from various directions.

\begin{proposition}\label{critical exp}
Let $m^*(K,J)$ be the unique maximizer of $\phi(m)$ defined in Corollary \ref{remark31}. Given $\alpha\in\mathbb{R}$ consider the lines 
\begin{equation}\label{J(K) curve}
J(K)=1+\alpha K\,\,,\,\,K>0
\end{equation}
and the function 
$m^*(K)\equiv m^*(K,J(K))$. Then, for $K\to 0^+$, the following holds

\begin{equation}
 m^*(K)\sim 
 \begin{cases}
 \sqrt{3\alpha} \sqrt{K}, & \text{for} \quad \alpha>0\\
 3K,& \text{for} \quad  \alpha=0 \\
 0 , &  \text{for} \quad \alpha<0.
 \end{cases}
 \end{equation}

\end{proposition}

\begin{remark} Notice that when $\alpha<0$ the critical exponent is 0.
The case  $K=0$ and $J\to 1^+$ corresponds to the classical Curie-Weiss model and is well known that  
\begin{equation}
    m^*(0,J)\sim \sqrt{\frac{3(J-1)}{J^3}}. 
\end{equation}
\end{remark}
 
\section{Proofs}\label{sec3:proofs}
This section contains the proofs of the above results and  is organised as follows:\\

In Section \ref{sec2},  we prove Proposition \ref{prop uniq} by studying the  properties of the function $\phi(m)$ appearing in the variational problem \eqref{varpressure}.
Section \ref{proof theorem} contains the proof of Theorem \ref{asymptotic_mag}   and is based on the asymptotic expansion given in Appendix \ref{appendix:B}.
Finally, in Section \ref{proof critical exp}, we derive the critical exponents of the model.

\subsection{Proof of Proposition \ref{prop uniq}}
\label{sec2}
The complete proof of Proposition \ref{prop uniq} follows from Propositions \ref{classification_m}, \ref{regularity}, \ref{wall} and \ref{regularity_wall} below. 

Let us start studying   in detail  the variational principle \eqref{varpressure} and observe that the function  $\phi(m)$ 
satisfies
\begin{equation}\label{derivative_VP}
\begin{split}
\frac{\partial }{\partial m}\phi(m)  &= Km^2 + Jm  - \frac{1}{2}\log\left(\frac{1+m}{1-m}\right),\cr
\frac{\partial^2 }{\partial m^2}\phi(m)    &= 2Km +J-\frac{1}{1-m^2}.
\end{split}
\end{equation} 
Therefore the  variational pressure $\phi(m)$ attains it maximum in at least one point $m=m(K,J) \in (-1,1)$, which satisfy

\begin{equation}\label{consistency eqn}
    \frac{\partial}{\partial m}\phi(m)=0, \qquad \text{i.e.,} \quad m=\tanh{(Km^2 + Jm)}.
\end{equation}
Indeed, from \eqref{derivative_VP} $\lim_{m\rightarrow-1^+} \phi'(m) = +\infty$ and  $\lim_{m\rightarrow1^-} \phi'(m) = -\infty$. Therefore, there exists $\epsilon>0$ such that $\phi(m)$ is strictly increasing on $[-1,-1+\epsilon]$ and strictly decreasing on $[1-\epsilon,1]$. This implies that, the local maximizers of $\phi(m)$ does not include $-1$ and $+1$. Notice also that, since $K>0$, if $\bar{m}>0$ then $\phi(\bar{m})>\phi(-\bar{m})$ therefore the supremum of $\phi(m)$ cannot be reached at negative values. 

A complete classification of the critical points of $\phi(m)$ is contained in the following proposition:

\begin{proposition}\label{classification_m}
(Classification of critical points) For all $K>0$ and $J\in\mathbb{R}$, the solutions to equation \eqref{consistency eqn} can be described as follow:\\

   Define the function
\begin{equation}\label{psik}
    \Psi(K) := \min_{m\in[0,1]} \frac{g(m,K)}{m} < 1 
\end{equation}
where $g(m,K):=\mathrm{arctanh}(m)-Km^2$ and set $J_c=1$. Then:
    \begin{itemize}
        \item[a.] for $J<\Psi(K)$, there exists a unique solution, $m_0 = 0$, and it is the maximum point of $\phi(m)$,
        \item[b.] for $\Psi(K)<J<J_c$, equation \eqref{consistency eqn} has three solutions, i.e., $m_0, m_1>m_3>0$. Furthermore, $m_0, m_1$ are local maximum points while $m_3$ is a local minimum point of $\phi(m)$, 
        \item[c.] for $J=\Psi(K)$, there exist two solutions, $m_0$ and $m_1>0$. Where $m_0$ is the maximum point of $\phi(m)$ and $m_1$ is an inflection point. 
        \item[d.] If $J\geq J_c$, there exists a unique positive solution $m_2$ which is the only maximum point of $\phi(m)$ in equation \eqref{varpressure}.
    \end{itemize}
\end{proposition}

\begin{proof}
Let us start by noticing that $m=0$ is always a solution of \eqref{consistency eqn}. Moreover, 

\begin{equation*}
    \phi''(0)
    \begin{cases}
      < 0, & \text{if $J<1$}\\
      > 0, & \text{if $J>1$}.
    \end{cases}       
\end{equation*}
Now, let's rewrite \eqref{consistency eqn} as

\begin{equation}\label{consistency_g}
    m J=\underbrace{\Bigg[\text{arctanh}(m)-Km^2\Bigg]}_{=g(m,K)}.
\end{equation}

The solutions of \eqref{consistency eqn} are the intersections between the line $mJ$  and the function  $g(m,K)$. Therefore the function $\Psi(K)$ in \eqref{psik} is a benchmark to study the number of  solutions of $\phi'(m)=0$ when $J$ varies. Indeed by definition, $\Psi(K)$ represents the smallest value of $J$ in order to have a positive solution for \eqref{consistency_g}. Let us start collecting some properties of the function 
$g(m,K)$. By definition we have  that 
\begin{equation}
\begin{split}
    g'(m,K) &= \Bigg[\frac{1}{1-m^2}-2Km\Bigg]\cr
    g''(m,K)&= \Bigg[\frac{2m}{(1-m^2)^2}-2K\Bigg].
\end{split}
\end{equation}
This implies that,

\begin{equation*}
   \begin{cases}
   g'(0,K)=1, \\
  g''(0,K)=-2K<0 \quad \text{for all  $K>0$}.
   \end{cases}
\end{equation*}
Since the function $m \mapsto \dfrac{2m}{(1-m^2)^2}$  is strictly increasing on $[0,1)$, then  $g''(m,K)=0$ has only one solution, namely   $g(m,K)$ has only one inflection point. Moreover, observe that, as $m\rightarrow 1^-$, $g(m,K) \rightarrow +\infty$.

\begin{enumerate}
    \item[a.] If $J<\Psi(K)$ then it's clear that \eqref{consistency eqn}  has a unique solution $m_0=0$ which is a maximum point since in this case $\phi''(0)<0$ .
    
    \item[b.] If $\Psi(K)<J<J_c$, continuity of $g$ and the fact that for $m\rightarrow 1^-$, $g(m,K) \rightarrow +\infty$, imply that \eqref{consistency eqn} has three solutions, $m_0, m_1$ and $m_3$, where $m_1$ and $m_3$ are positive. It's also easy to check using the properties of the function $g(m,K)$ that $m_0$ and $m_1$ are local maxima while $m_3$ is a local minima.
    
    \item[c.] If $J=\Psi(K)$, then there is only one intersection point $m_4$ between the line $mJ$ and the function $g(m,K)$. Standard reasoning  allows to conclude that  $m_4$  is an inflection point for $\phi$. 
    
    \item[d.] Finally suppose that $J\geq J_c$.  The fact that $ g'(0,K)=1$ and  $g''(0,K)=-2K<0$ for $K>0$, means that the line $mJ$ starts above the function $g$. Now, since $g$ has at most one inflection point and $g(m,K) \rightarrow +\infty$ as $m\to 1^-$, one can conclude that there exist a unique positive solution $m_2\in (0,1)$ of $\phi'(m)=0$. 
\end{enumerate}

\end{proof}

The solutions made mention in Proposition \ref{classification_m} are displayed in Figure \ref{gsolutions}.

\begin{figure}[ht]
    \centering
    \includegraphics[width=12cm]{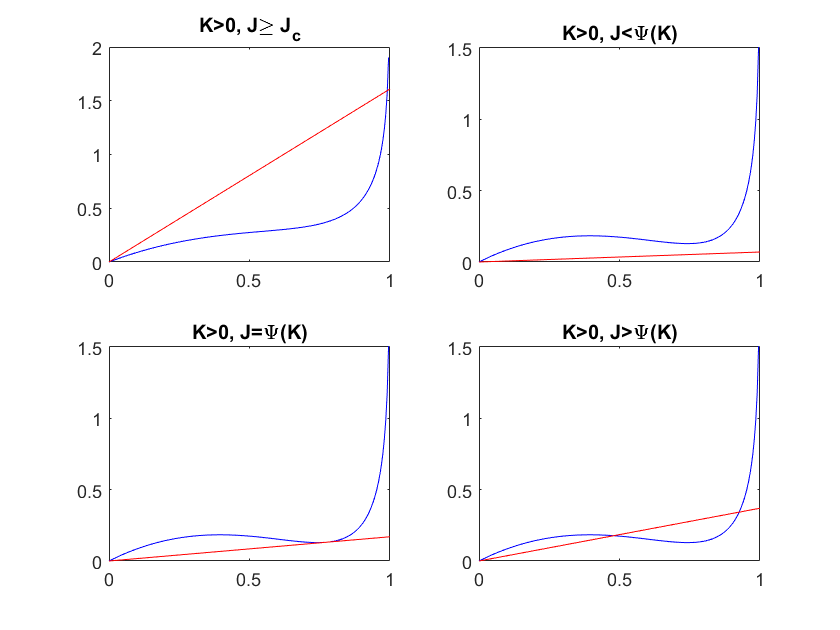}
    \caption{The points of intersection between the blue curve $g(m,K)$ as defined in \eqref{consistency_g} and red curve $f(m)=Jm$. The solution of the  equation \eqref{consistency eqn} are the points of intersection between $g(m,K)$ and $f(m)$.}
    \label{gsolutions}
\end{figure}

In the next proposition we obtain the differentiability of the solution(s) of the consistency equation \eqref{consistency eqn}  with respect to the parameters $J \; \text{and}\; K$.   

\begin{proposition}\label{regularity} (Regularity properties). Let  $m_0,m_1$ and $m_2$ be the (local) maxima  of $\phi$ described in Proposition \ref{classification_m}. Then for $K>0$, the following properties hold:

\begin{itemize}
    \item[(a)] $m_1$ is continuous in its domain namely $\Psi(K)\leq J<J_c$ and $C^\infty$ in its interior, while $m_2$ is  $C^\infty$ in  its domain, namely $J\geq J_c$.
   
   \item[(b)]  
   $\phi''(m_0)=\phi''(0)<0$, 
   $\phi''(m_1)<0$ for $\Psi(K)< J<J_c$, and
     $\phi''(m_2)<0$  for $J\geq J_c$.
\end{itemize}
Moreover, for any $i\in\{0,1,2\}$ it holds that

\begin{equation}
 \frac{\partial}{\partial J} \phi(m_i)=\frac{1}{2}m_i^2, \qquad \frac{\partial}{\partial K} \phi(m_i)=\frac{1}{3}m_i^3
\end{equation}

\begin{equation}
 \frac{\partial m_i}{\partial J} =-\frac{m_i}{\phi''(m_i)}, \qquad \frac{\partial m_i}{\partial K} =  -\frac{m_i^2}{\phi''(m_i)} .
\end{equation}
\end{proposition}
\begin{remark}
Notice that $(b)$ implies that there are no degenerate maximum points of $\phi(m)$ for $K>0$. Therefore the only degenerate maximum is  obtained for $(K,J)=(K_c,J_c)=(0,1)$, that is the critical point of a Curie-Weiss model, here the magnetization takes the value $m_c=0$.
\end{remark}

\begin{proof}
\textit{(a)} Let's start  with $m_1$ 
and  take  $(K,J)$ in its  domain, namely  $D:=\{(K,J)|K>0, \Psi(K)\leq J<J_c\}$  . We define  $\tau(K,J)= \bigg(\dfrac{1}{J}-1\bigg)\dfrac{J}{K}>0$ and $\tilde{\phi}(m):=\phi(m)|_{[\tau(K,J),1]}$. Observe from \eqref{consistency eqn} that,
    
\begin{equation*}
\begin{split}
    m_1=&\frac{1}{J}\Bigg[\underbrace{\text{arctanh}(m_1)}_{\geq m_1}-Km_1^2\Bigg]\cr
 \Longrightarrow m_1   \geq&\bigg(\frac{1}{J}-1\bigg)\frac{J}{K}=\tau(K,J).
\end{split}
\end{equation*}
Hence,  $m_1$ is the unique maximum point of  $\tilde{\phi}(m)$, then  by the Berge's maximum theorem \ref{prop B1} (see \cite{Ok2007, AlbericiContMing2014}), $m_1$ is continuous for  $(K,J)\in D$. To prove the smoothness of $m_1$ on the interior of its domain it's enough to show that $\phi''(m_1)<0$ and then apply the implicit function theorem \ref{prop B2} (see \cite{AlbericiContMing2014, Rudin1976}). Let $G(m):=\phi\,''(m)$ then,  

\begin{equation*}
\begin{split}
    \frac{\partial G}{\partial m}(m) &= 2K-\frac{2m}{(1-m^2)^2}\cr
    \frac{\partial^2G}{\partial m^2}(m) &= - \frac{2(3m^2+1)}{(1-m^2)^3}<0 \quad \forall \; m \in [0,1)
\end{split}
\end{equation*}
and hence,

\begin{equation}\label{propG}
   \begin{cases}
   G(0)=J-1<0, \quad \forall  \; J<J_c\\
  G\,'(0)= 2K > 0,  \quad \forall \; K>0\\
  G\,''(0)= -2 \\
  \lim_{m\rightarrow1^-} G(m) = -\infty \quad \forall \; K>0 \; \text{and} \; J<J_c  .
   \end{cases}
\end{equation}
We want to prove that $G(m_1)<0$ if $\Psi(K)<J<J_c$. Clearly since $m_1$ is a local maximizer it's enough to show that $G(m_1)\neq0$.
Recall that $m_1$ is the biggest positive solution of $\phi'(m)=0$. It's easy to check  that $G(m)=0$ has at most two solutions. Assume by contradiction that $G(m_1)=0$ if $\Psi(K)<J<J_c$, then  $G(m)<0$ or $G(m)>0$   in a left neighbourhood of $m_1$.

\begin{itemize}
\item Suppose that $G(m)<0$ in a left neighbourhood of $m_1$ then $G(m)$ cannot  be always negative, otherwise $\phi'(m)$ is decreasing and, since $\phi'(0)=0$ then $\phi'(m)=0$ can not have more than one solution. This contradicts point $b)$ of Proposition \ref{classification_m}. Therefore  there exist an interval where $G(m)>0$ but keeping in mind the properties of $G$ in \eqref{propG} and the fact that $G$ is continuous, this implies that there are at least  three solutions for $G(m)=0$, but this is impossible because we already observed that $G(m)=0$  has at most two solutions. 

\item Suppose that $G(m)>0$ in a left neighbourhood of $m_1$, then  $G(m)=0$ has in addition to $m_1$ another solution  that we denote by $\bar{m}$. Clearly  $\bar{m}<m_3$ otherwise $m_3$ cannot satisfies   $\phi'(m_3)=0$. Therefore $G(m)\equiv \phi''(m)>0$ if $m_3<m<m_1$ and this contradicts the fact that $\phi'(m_3)=\phi'(m_1)=0$.
\end{itemize}

Let's focus on $m_2$. Since for $K>0$ and $J\geq J_c$, $m_2$ is the only maximizer of $\phi(m)$ it's enough to show that $\phi''(m_2)<0$
to get smoothness of $m_2$ by using the implicit function theorem.  Let's note that if $J\geq J_c$ then  $\phi''(0)\geq0$  and   $\phi''(m)=0$ has a unique positive solution. Furthermore, $\phi(m)$ has a unique maximum point, $m_2\in (0,1)$ and $\phi'(m_2)= 0$. It is easy to show that   $\phi''(m_2)\neq 0$ by contradiction. Let's assume that $\phi''(m_2)= 0$ then $\phi''(m)> 0$ for $m<m_2$, therefore,  using the Taylor's series expansion of $\phi(m)$ around $m_2$ one gets $\phi(m)>\phi(m_2)$ which contradicts the fact that $m_2$ is the global maximum. 

Therefore by the implicit function theorem \ref{prop B2}, since  $\phi''(m)\neq0$ on the interior of the domains of $m_1$ and $m_2$, we can conclude that $m_1$ and $m_2$ are $C^\infty$. 
\vspace{1cm}

\textit{(b)} We already proved that for any $i\in \{0,1,2\}$, $\phi''(m_i)<0$ for suitable $K,J$. For the second part a direct computation shows that:

\begin{equation}\label{regJ_mu}
\begin{split}
    \frac{\partial}{\partial J}\phi(m_i) &=
    \frac{\partial}{\partial m}\phi(m)\bigg|_{m=m_i}\frac{\partial m_i}{\partial J}+\frac{m_i^2}{2}\cr
    &=\frac{m_i^2}{2}
\end{split}
\end{equation}
and similarly, 

\begin{equation}\label{regK_mu}
\begin{split}
    \frac{\partial}{\partial K}\phi(m_i)&=
    \frac{\partial}{\partial m}\phi(m)\bigg|_{m=m_i}\frac{\partial m_i}{\partial K}+\frac{m_i^3}{3}\cr
    &=\frac{m_i^3}{3}. 
\end{split}
\end{equation}
Using the fact that $m_i, i=\{0, 1,2\}$ are the stationary points of  $\phi(\cdot)$, we have that $\dfrac{\partial m_i }{\partial K}$ satisfies

\begin{equation}
\begin{split}
     &\frac{1}{1-m_i^2} \frac{\partial m_i}{\partial K} - m_i^2 - 2Km_i\frac{\partial m_i}{\partial K} - J\frac{\partial m_i}{\partial K} = 0 \cr
    & \frac{\partial m_i}{\partial K} \Bigg[\frac{1}{1-m_i^2}-2Km_i -J\Bigg] =m_i^2\cr
    &\frac{\partial m_i}{\partial K} = -\frac{m_i^2}{\phi''(m_i)}
\end{split}
\end{equation}
and similarly for $\dfrac{\partial m_i }{\partial J}$ one obtains

\begin{equation}
\begin{split}
   & \frac{1}{1-m_i^2} \frac{\partial m_i}{\partial J} - 2Km_i\frac{\partial m_i}{\partial J} - m_i - J\frac{\partial m_i}{\partial J} = 0 \cr
     &\frac{\partial m_i}{\partial J} \Bigg[\frac{1}{1-m_i^2}-2Km_i -J\Bigg] = m_i\cr
    &\frac{\partial m_i}{\partial J} = -\frac{m_i}{\phi''(m_i)}
\end{split}
\end{equation}
and this concludes the proof.
\end{proof}

Now we study which of the stationary points described by Proposition \ref{classification_m} are global maximizers of $\phi(m)$ and show the existence of a phase transition. These stationary points are: $m_0, m_1$, and  $m_2$. Let us start by recalling the result of  Proposition \ref{classification_m}:   
\begin{itemize}
\item if $J< \Psi(K)$, then  $m_0$ is the only global maximum point of $\phi$

\item 
if $\Psi(K)<J<J_c$ then  $\phi(m)$ has two local maximizers $m_0$ and $m_1$ 

\item if $J\geq J_c$ then $m_2$ is the only the global maximum point of $\phi(m)$
\end{itemize}

To identify the coexistence of two global maximum points of $\phi(m)$ when $\Psi(K)<J<J_c$, consider the following function:
\begin{equation}\label{Delta}
    \Delta(K,J)=  \phi(m_1,K,J) - \phi(m_0, K,J).
\end{equation}
Notice that $\Delta(K,J)$ can be extended by continuity at $J=\Psi(K)$ and $J=J_c$. In the above equation we use $\phi(\cdot,K,J)$ to emphasis the dependence of $\phi$ on the parameters.

\begin{proposition}\label{wall}
(Existence and uniqueness). For all $K>0$ there exists a unique $J=\gamma(K)\in (\Psi(K), J_c)$ such that $\Delta(K,J) = 0$. Furthermore,

\begin{equation}
    \Delta(K,J)
   \begin{cases}
   < 0 , & \text{if \; $\Psi(K)\leq J < \gamma(K)$ }\\
   > 0 , & \text{if \; $ \gamma(K)<J\leq J_c$}.
   \end{cases}
\end{equation}
\end{proposition}

\begin{proof}
Let us start by observing that 
\begin{itemize}
    \item $\Delta(K,\Psi(K))<0$, since for $J=\Psi(K)$, $m_0$ is the only maximum point of $\phi(m,K,J)$.
    
    \item $\Delta(K,J_c)>0$, since $\lim_{J\rightarrow 1^-} m_1(K,J)=m_2(K,1)$ and  $m_2(K,1)$ is the only global maximum for $\phi(m,K,J)$.
    \end{itemize}
    
Now, by continuity of $\phi(m)$ and $ m_1$, we have that $J\mapsto \Delta(K,J)$ is a continuous function, and then the existence of the \textit{wall} $J=\gamma(K)$ follows  from the application of the intermediate value theorem. For the uniqueness part we observe that  $J\mapsto \Delta(K,J)$ is strictly increasing. Indeed  from  Proposition \ref{regularity} we know that  $\phi(m_1), m_1$ are smooth functions and

\begin{equation}
\begin{split}
    \frac{\partial \Delta}{\partial J}(K,J) &= \frac{\partial}{\partial J}\phi(m_1) - \frac{\partial}{\partial J}\phi(m_0)\cr
    &= \frac{1}{2}m_1{^2} - \frac{1}{2}m_0{^2} \cr
    &= \frac{1}{2}m_1{^2} >0
\end{split}
\end{equation}
for $J\in (\Psi(K),J_c)$. 
\end{proof}

\begin{corollary}\label{remark31}
The  function $\phi(m)$  has a unique global maximum point $m^*(K,J)$  given by:

\begin{equation}\label{Remark m^*}
    m^*(K,J) :=
   \begin{cases}
   m_0=0 , & \text{if \; $J<\gamma(K)$}\\
   m_1(K,J) , & \text{if\; $\gamma(K)<J<J_c$},\\
   m_2(K,J) , & \text{if \; $J\geq J_c,\;$ }\\
   \end{cases}
\end{equation}
where the function $\gamma(K)$ is defined by Proposition \ref{wall} and $\phi''(m^*)<0$. 
\end{corollary}
Note that on the curve $\gamma$ there are two global maximum points of $\phi(m)$.  Let us define

\begin{equation}
    \overline{\gamma}(K) :=
    \begin{cases}
    \gamma(K) , & \text{if \; $K > 0$ }\\
   J_c , & \text{if \; $K=K_c=0$}.
    \end{cases}
\end{equation}
Therefore by Proposition \ref{regularity} one can conclude that  $m^*(K,J)$ is continuous on its domain $(\mathbb{R}_+\times\mathbb{R})$\textbackslash $\gamma $ and it is $C^\infty$ on $(\mathbb{R}_+\times\mathbb{R})\setminus \overline{\gamma} $. Moreover the following holds:

\begin{proposition}\label{regularity_wall}
(Regularity properties.) The function $\overline{\gamma}(K)$ is $C^\infty(\mathbb{R}_+\setminus\{0\})$  and at least $C^1$  for $K=0$. In particular,

\begin{equation}
    \gamma\;'(K) := -\frac{2}{3}m_1(K,\gamma(K)) \qquad \forall \quad K>0
\end{equation}

and 
\begin{equation}
    \overline{\gamma}\; '(K_c) := -\frac{2}{3}m_c.
\end{equation}
\end{proposition}

\begin{proof}
i. We begin by showing that $\gamma(K)\in C^\infty(\mathbb{R}_+)$.  By Proposition \ref{wall}, $J=\gamma(K)$ is a unique solution of the equation 
\begin{equation*}
    \Delta(K,J) = 0,
\end{equation*}
where $\Delta$ is defined by equation \eqref{Delta} for $\Psi(K)\leq J<J_c$ and $K>0$. Furthermore, observe that $\Delta$ is $C^\infty$ in its domain  by the smoothness of $\phi$ and $m_1$. Recall from the proof of Proposition \ref{wall} that 

\begin{equation}
    \frac{\partial }{\partial J}\Delta(K,J) \neq 0 \qquad \forall \quad (K,J) \, \text{s.t.} \;\; J=\gamma(K), 
\end{equation}
hence, by the implicit function theorem \ref{coro B1} $\gamma(K)\in C^\infty(\mathbb{R}_+)$. Therefore
\begin{equation}
\begin{split}
    \Delta(K,\gamma(K)) \equiv 0 &= \frac{d}{dK}\Delta(K,\gamma(K))\cr
    &=  \frac{\partial \Delta}{\partial J}(K,\gamma(K)) \gamma\,'(K) + \frac{\partial \Delta}{\partial K}(K,\gamma(K))\cr
    \Longrightarrow \gamma\,'(K) =& -\frac{\partial \Delta}{\partial K}/\frac{\partial \Delta}{\partial J}(K,\gamma(K))
\end{split}
\end{equation}
From equations \eqref{regJ_mu} and \eqref{regK_mu}, we have that,
\begin{equation*}
    \frac{\partial \Delta}{\partial K} = \frac{m_1^3}{3} - \frac{m_0^3}{3} \qquad \text{and} \qquad \frac{\partial \Delta}{\partial J} = \frac{m_1^2}{2} - \frac{m_0^2}{2},
\end{equation*}
hence

\begin{equation}
    \gamma\,'(K) = -\frac{2}{3} m_1(K,\gamma(K))
\end{equation}
since $m_0(K,\gamma(K))=0, \; \forall \, K>0$.
Notice that, by \eqref{mean-field eqn}, $m_1(K,\gamma(K)) \xrightarrow[K\rightarrow \infty]{} 1$ which implies that $$\lim_{K\to \infty}\gamma\,'(K)=-\frac{2}{3}.$$ A consequence of this property is that also when $J<0$ (antiferromagnetic case) and very large there is always going to be phase transition between a polarized and unpolarized state.

ii. Now we prove that the extended function $\overline{\gamma}\in C^1(\mathbb{R}_{+})$. 
Recall that $\gamma(K)\in[\Psi(K),J_c]$
and observe that $\lim_{K\to K_c^+}\Psi(K)=J_c$ then
\begin{equation*}
    \lim_{K\rightarrow K_c^+} \gamma(K) = J_c 
\end{equation*}
which implies that $\overline{\gamma}$ is continuous at $K_c$. 
Now we have that \begin{equation}\label{gradient_wall}
    \gamma\,'(K) = -\frac{2}{3} m_1(K,\gamma(K)) \;
    \xrightarrow[K\rightarrow K_c^+]{} -\frac{2}{3} m_c=0
\end{equation}
which implies that $
\overline{\gamma}'(K_c)=-\frac{2}{3} m_c=0$ by the application of mean value theorem.
\end{proof}

\subsection{Proof of Theorem \ref{asymptotic_mag}}\label{proof theorem}
In this section we provide the details of the proof for Theorem  \ref{asymptotic_mag} following closely the argument in \cite{Mukherjee2021}.

\begin{proof}

$1.$ By proposition  \ref{prop uniq} we know that if $(K,J)\in (\mathbb{R}_+\times\mathbb{R}) \backslash (\gamma\cup (K_c,J_c))$ then $\phi(m)$ has a unique global maximizer $m^*$ with $\phi''(m^*)<0$. It's easy to check that $\phi(m)$ satisfies the hypothesis of Lemma  \ref{Appendix: lemma ZN exp unique}, therefore   \eqref{conc1} gives concentration inequality for $m_N$ in a suitable  neighbourhood of $m^*$ under the probability measure \eqref{Gibbs}. More precisely, for any  $\alpha\in(0,\frac{1}{6}]$ and $N$ large enough one has

\begin{equation}\label{conc1p}
\mu_N(m_N\in B^c_{N,\alpha}(m^*))= \exp\bigg\{\frac{1}{2}N^{2\alpha} \phi''(m^*)\bigg\}\mathcal{O}(N^{\frac{3}{2}})
\end{equation}

where $B^c_{N,\alpha}(m^*)=\{m\in\mathbb{R}:|m-m^*|\leq N^{-\frac{1}{2}+\alpha}\}$. Therefore the convergence in distribution \eqref{lln_m*} follows from \eqref{conc1p} by standard approximation arguments.

To obtain the central limit for  $m_N$, it is enough to compute the limit of the moment generating function of the rescaled random variable $N^{\frac{1}{2}}(m_N-m^*)$.   For a fixed $t\in\mathbb{R}$,
the moment generating function of $N^{\frac{1}{2}}(m_N-m^*)$ can be expressed as 

\begin{equation}\label{mgf}
    \mathbb{E}\bigg[e^{tN^{\frac{1}{2}}(m_N-m^*)} \bigg] = e^{-tN^{\frac{1}{2}}m^{*}} \frac{\bar{Z}_N(t)}{Z_N}.
\end{equation}

where $\bar{Z}_N(t)$ is a perturbed partition function associated to an Hamiltonian

\begin{equation}
\bar{H}_N(\sigma)= H_N(\sigma)+ \sqrt{N}\,t\,m_N(\sigma).
\end{equation}

We start by noticing that $\bar{H}_N(\sigma)= -N f_N(m_N(\sigma))$ where $f_N(x)= \frac{K}{3} x^3 +\frac{J}{2} x^2  +\frac{1}{\sqrt{N}}t x$ and then $f_N$ together with all its derivatives tends uniformly to $f(x)= \frac{K}{3} x^3 +\frac{J}{2} x^2$. Therefore one can use  Lemma \ref{Appendix: lemma ZN exp unique} to obtain an asymptotic expansion for both $Z_N$ and $\bar{Z}_N(t)$. More precisely one gets 
\begin{equation}\label{zexpansion}
    \frac{\bar{Z}_N(t)}{Z_N} = e^{N\big(\phi_N(m_N^*(t))-\phi(m^*)\big)} (1+\mathcal{O}(N^{-\frac{1}{2}+3\alpha})),
\end{equation}

where $\phi_N(x)=f_N(x)-I(x)$ and for $N$ large enough $m_N^*(t)$ is its unique maximizer. Now, let's observe that $m_N^*(0) = m^*$ and  $m_N^*(t)$ satisfies the equation

\begin{equation}
m_N^*(t) = \tanh \bigg(Km_N^*(t)^2+Jm_N^*(t)+\frac{t}{\sqrt{N}}\bigg).
\end{equation}

Hence, it's easy to check that $\frac{\partial m_N^*(t)}{\partial t}|_{t=0} = -\frac{1}{\sqrt{N}\phi ''(m^*)}$ and  $\frac{\partial^2 m_N^*(t)}{\partial t^2}= \mathcal{O}(N^{-1})$. Therefore the Taylor's expansion of $m_N^*(t)$ around $t=0$ is
\begin{equation}\label{mexpansion}
    m_N^*(t) = m_N^*(0) -\frac{t}{\sqrt{N}\phi ''(m^*)} +  \mathcal{O}(N^{-1}).
\end{equation}
Moreover one can easily check that $\phi_N(m_N^*(t))=\phi(m_N^*(t))+\dfrac{t}{\sqrt{N}}m_N^*(t)$. Hence the  Taylor expansion of  $\phi(m_N^*(t))$ around $m^*$ is
\begin{equation}\label{phiN-phi}
    N\bigg(\phi_N(m_N^*(t))-\phi(m^*) \bigg) = \frac{N}{2}\bigg[(m_N^*(t))-(m^*)\bigg]^2 \phi ''(m^*)+ \sqrt{N}tm_N^*(t) +o(1).
\end{equation}
Finally using equations  \eqref{mexpansion} and \eqref{phiN-phi} in the above, one gets

\begin{equation}
N\bigg(\phi_N(m_N^*(t))-\phi(m^*) \bigg) =   t\sqrt{N} m^{*} - \frac{t^2}{2\phi''(m^*)} +o(1)
\end{equation}
and by \eqref{zexpansion} the limiting moment generating function is given as
\begin{equation}\label{clt_unique}
    \lim_{N\rightarrow\infty} \mathbb{E}\bigg[e^{tN^{\frac{1}{2}}(m_N(\sigma)-m^*)} \bigg] = \exp{\bigg\{-\dfrac{t^2}{2\phi''(m^*)}\bigg\}},
\end{equation}
which implies \eqref{asymptotic_uniq}.

$2.$ Let's recall that by Proposition \ref{prop uniq} there exist two global maximizers $m_i$ of $\phi(m)$ for $i\in\{0,1\}$ on $\gamma$. Moreover by point $b)$  of Proposition \ref{regularity} we know that $\phi''(m_i)<0$ for  $i\in\{0,1\}$. Now, following the same argument as before, formula \eqref{gibbs mi} in Lemma \ref{Appendix: lemma ZN exp many max} gives the concentration inequality for $m_N$ within a suitable neighbourhood of $m_i$ with respect to the Gibbs measure \eqref{Gibbs}. Therefore the convergence in distribution \eqref{lln_mi} and \eqref{rho_i mi} follows the asymptotic expansions of the (restricted) partition function  in  Lemma \ref{Appendix: lemma ZN exp many max}.

To obtain the local central limit theorem for $m_N$  around the global maximizers $m_i$, we will show that the moment generating function of $N^{\frac{1}{2}}(m_N-m_i)\big|\{m_N\in A_i\}$ with respect to the measure $\mu_N$ converges pointwise in distribution to the moment generating function of $\mathcal{N}\bigg(0,-\frac{1}{\phi''(m_i)}\bigg)$. Here $A_i \subset [-1,1]$ is such that $m_i$ is the unique maximizer of $\phi(m)$ on its interior. The moment generating function of $N^{\frac{1}{2}}(m_N-m_i)\big|\{m_N\in A_i\}$ at a fixed $t\in \mathbb{R}$ is

\begin{equation}\label{mgf mi}
    \mathbb{E}\bigg[e^{tN^{\frac{1}{2}}(m_N-m_i)}\bigg|\{m_N\in A_i\} \bigg] = e^{-tN^{\frac{1}{2}}m_i} \frac{\bar{Z}_N(t)\big|_{A_i}}{Z_N\big|_{A_i}}.
\end{equation}
Following the asymptotic expansion of the partition function in \eqref{Z|A mi} (see Lemma \ref{Appendix: lemma ZN exp many max}), the fraction on the right side of equation \eqref{mgf mi} reduces to 

\begin{equation}
    \frac{\bar{Z}_N(t)\big|_{A_i}}{Z_N\big|_{A_i}} \sim e^{N\big(\phi_N(m_{i,N}(t))-\phi(m_i)\big)}.
\end{equation}
Now, taking Taylor's expansion of $\phi_N(m_{i,N}(t))$ at $m_i$ up to the second order, one can repeat the same arguments as in the unique maximum case, obtaining 
\begin{equation}
    \mathbb{E}\bigg[e^{tN^{\frac{1}{2}}(m_N-m_i)}\bigg|\{m_N\in A_i\} \bigg] \xrightarrow[N\rightarrow\infty]{} \exp\bigg\{-\frac{t^2}{2\phi''(m_i)}\bigg\}.
\end{equation}
This completes the proof of \eqref{asymptotic_mi}.

$3.$ Notice that the critical point $(K_c,J_c)=(0,1)$ is a degenerate maximum point for $\phi(m)$ in the sense that $\phi''(m^*(K,J))\big|_{(K,J)=(0,1)}=0$. This does not allow the use of the asymptotic expansions in Lemma \ref{Appendix: lemma ZN exp unique}.
However, one can simply  notice that the Hamiltonian $H_N$ of the model  at the critical point $(K_c,J_c)=(0,1)$ coincides at any $N\in\mathbb{N}$ with the Hamiltonian function of the standard Curie-Weiss model at the critical temperature $J=1$ and zero external field. Therefore  \eqref{lln_mc}  and \eqref{asymptotic_cpt} are a well known results and their proof can be found in   \cite{EllisNewman78}.

\end{proof}


\subsection{Proof of Proposition \ref{critical exp}}\label{proof critical exp}

\begin{proof}
Let us start with the case $\alpha\geq0$. This implies from equation \eqref{J(K) curve} that $J(K)\geq J_c=1$ and then $m^*(K)\equiv m_2(K,J(K))$ where $m_2$ is the only positive solution of the consistency equation \eqref{consistency eqn}.

Clearly $m^*(K)\to 0$ as $K\to 0^+$, hence  by Taylor's expansion we have that

\begin{equation}
\begin{split}
     m^*(K) &= J(K)m^*(K) + Km^*(K)^2 - \frac{J(K)^3m^*(K)^3}{3} + \mathcal{O}(m^*(K)^4)\cr
     &= (1+\alpha K)m^*(K) +Km^*(K)^2 - \frac{(1+\alpha K)^3m^*(K)^3}{3}+\mathcal{O}(m^*(K)^4).
\end{split}
\end{equation}
Hence 
\begin{equation*}
   \frac{(1 +\alpha^3K^3+3\alpha^2K^2+3\alpha K)m^*(K)^2}{3} - K m^*(K) -\alpha K= \mathcal{O}(m^*(K)^3).
\end{equation*}
From the above equation, neglecting higher order corrections  we have 

\begin{equation}\label{calc_crit exp}
    m^*(K) \sim \frac{3}{2} \left(K+\sqrt{K^2+ \frac{4}{3}\alpha K + 4\alpha^2K^2}  \right) .
\end{equation}
Now, if $\alpha > 0$ then  

\begin{equation}\label{}
    m^*(K)\sim \sqrt{3\alpha K}.
\end{equation}
Otherwise if $\alpha = 0$, then 

\begin{equation}
    m^*(K)\sim \frac{3}{2} \left(K+ \sqrt{K^2}  \right) \sim 3 K.
\end{equation}
 
Let's turn on the case  $\alpha < 0$. From Proposition \ref{regularity_wall} we know that $\gamma(K)$ is at least $C^1$ at $K=0$. Since $\lim_{K\to 0^+} \gamma\,'(K)=0$ we know that if $J(K)<\gamma(K)$ for $K$ small enough, then $m^*(K)\equiv m_0(K,J)=0$.

\end{proof}

\section{Conclusion and perspectives}\label{sec4:conclusion}
In this work, we have studied how the three-body interaction, which provides a spin-flip symmetry-breaking parameter, induces  phase transitions with novel properties in the mean-field setting. In particular, we derived all the critical exponents and the limiting distribution of a suitably rescaled magnetization in the entire phase space. The presence of a stable paramagnetic phase and the fact that, also in the antiferromagnetic regime, the model presents phase transitions and phase coexistence are interesting for applications in socio-technical environments \cite{ContucciKO2022} and possibly in other fields \cite{LiggettSteif2007, MajhiPercGhosh2022}. 

A possible research development will be to extend the results of the present work to multi-populated models \cite{ContucciKO2022, LoweSch2018, LoweSchVer2020, BerthetRS2019, GalloBC2009, GalloC2008, FedeleC2011, ContucciGG2008, OpokuOsa2019, ContucciGh2007}. In these models, the invariance of the Hamiltonian with respect to  permutations among sites is replaced by a weaker one that takes into account the existence of different species of spins. This setting is particularly useful in social science applications \cite{ContucciKO2022, GalloBC2009, ContucciGG2008, OpokuOsa2019}. Moreover, as mentioned in the introduction,  the mean-field approximation involved in the study of some finite-dimensional lattices provides a natural emergence of the multi-populated models.
 It is well known, for instance, that a system on a simple cubic lattice \cite{KinCohen1975, GalamYS1998} with ferromagnetic and antiferromagnetic couplings has a factorized equilibrium measure that corresponds to a two-populated mean-field model. Similarly, it has been shown in \cite{Bidaux1986} that on a regular square lattice, a system with cubic interaction has a product state equilibrium described by a two-populated mean-field model, while on a regular triangular lattice \cite{FroyenSH1976}, by a three-populated mean-field model.  
 
We also mention that in the case of quadratic interaction, Stein's method provides stronger results (Berry-Esseen type bounds) on the rate at which the convergence to the normal distribution takes place (see \cite{EichelsbacherLowe2010,ChatterjeeShao2011}). The extension of the above method to our model  and more generally to higher order interaction  is an interesting open problem. We plan to develop those research directions in the future.

\appendix

\section{Technical results}\label{appendix:A}
This section of the appendix presents some useful technical results applied in the work. We begin by stating the Berge's maximum theorem in the following Proposition.

\begin{proposition}\label{prop B1}
    Let $f:[-1,1]\times \mathbb{R}^n \to \mathbb{R}$ and $c:\mathbb{R}^m \to [-1,1]$ be continuous functions.
    \begin{itemize}
        \item[(a)] The following function is continuous:
        \begin{equation*}
            F: \mathbb{R}^n \times \mathbb{R}^m \to \mathbb{R}, \quad F(x,y) =\max_{v\in[-1,c(y)]} f(v,x).
        \end{equation*}

        \item[(b)] Suppose that for all $x,y\in \mathbb{R}^n$ the function $v\mapsto f(v,x)$ achieves its maximum on $[-1,c(y)]$ in a unique point. Then also the following function is continuous:
        \begin{equation*}
            V: \mathbb{R}^n \times \mathbb{R}^m \to [-1,1], \quad V(x,y) = \underset{v\in[-1,c(y)]}{\mathrm{argmax}} f(v,x).
        \end{equation*}
    \end{itemize}
\end{proposition}
The following proposition partially states Dini's implicit function theorem. Then we provide two simple corollaries that are used in the paper.

\begin{proposition}\label{prop B2}
    Let $F: \mathbb{R}^n \times \mathbb{R} \to \mathbb{R}$ be a $C^{\infty}$ function. Let $(x_0,y_0)\in \mathbb{R}^n \times \mathbb{R}$ such that $F(x_0,y_0)= 0$ and $\frac{\partial F}{\partial y} (x_0,y_0) \neq 0$. Then there exist $\delta>0, \epsilon>0$ and a $C^{\infty}$ function $f:B(x_0,\delta) \to B(y_0,\epsilon)$ such that for all $(x,y) \in B(x_0,\delta) \times B(y_0,\epsilon)$
    $$F(x,y)=0 \quad \iff \quad y = f(x)$$
\end{proposition}

\begin{corollary}\label{coro B1}
     Let $F: \mathbb{R}^n \times \mathbb{R} \to \mathbb{R}$ be a $C^{\infty}$ function. Let $\varphi: \mathbb{R}^n \to \mathbb{R}$ be a continuous
function such that for all $x\in\mathbb{R}^n$ such that $F(x,\varphi(x))= 0$ and $\frac{\partial F}{\partial y} (x,\varphi(x)) \neq 0$, then $\varphi(x)\in C^{\infty}(\mathbb{R}^n)$.
\end{corollary}

\begin{corollary}\label{coro B2}
     Let $F: \mathbb{R}^n \times \mathbb{R} \to \mathbb{R}$ be a $C^{\infty}$ function. Let 
     $a,b: \mathbb{R}^n \to \mathbb{R}$ be a continuous
function such that for all $a<b$. Suppose that for all $x\in \mathbb{R}^n$ there exists a unique $y=\varphi(x) \in (a(x),b(x))$ such that $F(x,\varphi(x))= 0$. Moreover, suppose that for all $x\in \mathbb{R}^n$, $\dfrac{\partial F}{\partial y} (x,\varphi(x)) \neq 0$, then $\varphi(x)\in C^{\infty}(\mathbb{R}^n)$.
\end{corollary}

\section{Concentration results and asymptotic expansions}\label{appendix:B} 
In this section of the appendix, we state concentration properties of the magnetization and asymptotic expansions of the partition function for a large class of Ising mean-field models and give proofs using the same methods and arguments recently introduced in \cite{Mukherjee2021}.

Consider a mean-field spin model with energy density $f_N$, namely
\begin{equation}
H_N(\sigma)=-Nf_N(m_N(\sigma))\,,\, \,\, \sigma\in \{ -1,1\}^N
\end{equation}

where $m_N=\frac{1}{N}\sum_{i\leq N} \sigma_i$ is the magnetization density. 
We assume that  $(f_N )$  is a sequence of continuous  functions $f_N:[-1,1]^N\to \mathbb{R}$ converging uniformly to $f$. We assume also that $f_N$ has bounded derivatives up to order 4  converging uniformly to $f', f'', f''', f''''$. We denote the law of the magnetization under  the Gibbs measure by
\begin{equation}\label{appendix gibbs}
    \mu_N(\sigma) = \dfrac{e^{-H_N(\sigma)}}{Z_N} .
\end{equation}
The partition function   $Z_N$ can be written as

\begin{equation}\label{A-partition}
Z_N= \sum_{x\in R_N} A_N(x) e^{N f_N(x)},
\end{equation}
where $R_N=\{-1+\frac{2k}{N},k=0,\ldots, N\}$ and $A_N(x)=\mathrm{card}\{\sigma\in \{-1,1\}^N: m_N(\sigma)=x \}$. Now, it follows from \cite{Talagrand2003} that, for some universal constant $L$
\begin{equation}\label{talagrand}
    \dfrac{1}{L\sqrt{N}}e^{-NI(x)}\leq A_N \leq e^{-NI(x)}
\end{equation}
where $I(x)$ is defined in \eqref{entropy}. Define the sequence  $\phi_N$ as 
\begin{equation}\label{Appendix free}
\phi_N(x)=f_N(x)-I(x).
\end{equation}
Notice the assumption on $(f_N)$ that  $\phi_N\to \phi = f-I$
uniformly on $(-1,1)$, as well as its derivarites up to order 4 on $(-1,1)$. 

The following lemmata contains  concentration properties of the magnetization $m_N$ w.r.t. the Gibbs measure $\mu_N$ and  asymptotic expansions of the partition function $Z_N$. For any $\alpha>0$ and $y\in \mathbb{R}$ we denote by $B_{N,\alpha}(y)$ the open ball with center $y$ and radius $N^{-1/2+\alpha}$ and by $B^c_{N,\alpha}(y)$ its complement.



\begin{lemma}\label{Appendix: lemma ZN exp unique}
Assume  that $\phi(x)$ has a unique global maximizer $x^*\in(-1,1)$ such that  $\phi''(x^*)<0$.  Then for $N$ large enough $\phi_N$  has a unique maximizer $x^*_N\rightarrow x^*$ such that $\phi_N''(x_N^*)<0$. Moreover for $\alpha\in\big(0,\frac{1}{6}\big]$ and $N$ large enough we have that

\begin{equation}\label{conc1}
\mu_N(m_N\in B^c_{N,\alpha}(x_N^*))= \exp\bigg\{\frac{1}{2}N^{2\alpha} \phi_N''(x^*_N)\bigg\}\mathcal{O}(N^{\frac{3}{2}})
\end{equation}

and the partition function \eqref{A-partition} can be expanded as,
\begin{equation}\label{Z m*}
    Z_N = \dfrac{e^{N\phi_N(x_N^*)}}{\sqrt{(x_N^{*^2}-1)\phi_N''(x_N^*)}}\bigg(1+\mathcal{O}\bigg(N^{-\frac{1}{2}+\alpha}\bigg) \bigg),
\end{equation}
\end{lemma}
\begin{proof}

Let  $x_N^*$ be any maximizer of $\phi_N$ which exists since $[-1,1]$ is compact.  Then there exist a subsequence $\{N_{l}\}_{l\geq 1}$   such that $x^*_{N_l}$ converges to some $y$. We know that $\phi_{N_l}(x^*_{N_l}) \geq \phi_{N_l}(x)$ for all $x\in[-1,1]$, therefore  by uniform convergence and taking  $ l\rightarrow\infty$  we obtain $\phi(y)\geq \phi(x)$ for all $x\in[-1,1]$ and this implies that $y$ is a global maximizer of $\phi(x)$. But $x^*$ is the unique global maximizer of $\phi(x)$, hence $y=x^*$. 

Since $\phi''(x^*)<0$ one has, for $\epsilon$ small enough, $\phi(x)<0$ for any  $x\in[x^*-\epsilon,x^*+\epsilon]$. Let $x_{N}$ and $y_N$ be two global maximizers of $\phi_N$. We already know that $x_N\to x^*$ and $y_N\to x^*$.
Therefore for $N$ large enough 
$x_N,y_N \in[x^*-\epsilon,x^*+\epsilon]$. Using the fact that $\phi''_N$ converges uniformly to $\phi''$ 
one can show that for $N$ large enough 
$\phi_N$ it is strongly convex on   $[x^*-\epsilon,x^*+\epsilon]$ and therefore has unique maximizer which implies that $x_N=y_N$. 

In order to lighten the notation set $B_{N,\alpha}=B_{N,\alpha}(x^*_N)$.    From equations \eqref{appendix gibbs}, \eqref{A-partition} and \eqref{talagrand} we have that,
\begin{equation}\label{prob_m*}
\begin{split}
    \mu_{N}(m_N\in B_{N,\alpha}^c) &= \dfrac{\sum_{x\in R_N\cap B_{N,\alpha}^c}A_N \exp\big\{N\big(f_N(x)\big)\big\}}{\sum_{x\in R_N}A_N \exp\big\{N\big(f_N(x)\big)\big\}}\cr
    & \leq \dfrac{LN^{\frac{1}{2}}(N+1)\sup_{x\in B_{N,\alpha}^c}e^{N\phi_N(x)}}{\sup_{x\in [-1,1]}e^{N\phi_N(x)}}\cr
    & = \exp{\bigg\{N\bigg(\sup_{x\in B_{N,\alpha}^c}\phi_N(x) - \phi_N(x_N^*) \bigg)\bigg\}} \mathcal{O}(N^{\frac{3}{2}}).
\end{split}
\end{equation}
Now by Lemma B.11 in \cite{Mukherjee2021} we know that for $x\in B_{N,\alpha}^c$  the maximizer of $\phi_N(x)$ is,  for $N$ large enough, either $x_N^*-N^{-\frac{1}{2}+\alpha}$ or $x_N^*+N^{-\frac{1}{2}+\alpha}$. This implies that $\sup_{x\in B_{N,\alpha}^c(x^*_N)}\phi_N(x)$ is either $\phi_N(x_N^*-N^{-\frac{1}{2}+\alpha})$ or $\phi_N(x_N^*+N^{-\frac{1}{2}+\alpha})$.   Note that 
$\phi'_N(x_N^*)=0$ since $x_N^*$ is the maximizer and  $\phi_N^{(3)}(x_N^*)$ is uniformly bounded on any closed interval in $(-1,1)$. Hence by a second-order Taylor expansion of $\phi_N(x_N^*\pm N^{-\frac{1}{2}+\alpha})$ at the point $x_N^*$, we have that

\begin{equation}
\phi(x_N^*\pm N^{-\frac{1}{2}+\alpha}) = \phi_N(x_N^*)+\frac{1}{2}N^{-1+2\alpha}\phi_N''(x_N^*)+\mathcal{O}(N^{-\frac{3}{2}+3\alpha}),
\end{equation}
where $\phi''(x_N^*)<0$. This completes the proof of equation \eqref{conc1}  following from equation \eqref{prob_m*}. 

To complete the proof of Lemma \ref{Appendix: lemma ZN exp unique}, let's start by observing that almost all the contribution to $Z_N$ comes from spin configurations having magnetization in a vanishing neighbourhood of the maximizer $x_N^*$, i.e., 
$\mu_{N}(m_N(\sigma)\in B_{N,\alpha})= 1-\mathcal{O}(e^{-N^\alpha})$. Hence,

\begin{equation}\label{z m uniq}
    Z_N =(1+\mathcal{O}(e^{-N^\alpha}))\sum_{x\in R_N\cap B_{N,\alpha} } \underbrace{
    \binom{N}{\frac{N(1+x)}{2}} \exp\big\{N\big(f_N(x)\big)\big\}.
    }_{=\zeta(x)}
\end{equation}
where $\zeta: [-1,1] \to \mathbb{R}$.
With this, one can accurately approximate the partition function over all configurations $\sigma$ whose mean lies within a vanishing neighbourhood of $x^*$ using standard approximation techniques. 

We begin by applying the  Laplace approximation of an integral over a shrinking interval $B_{N,\alpha}$ via the Riemann approximation of the sum in equation \eqref{z m uniq} with an integral and the binomial coefficient can be approximated by the Stirling's approximation method. 
Notice that by the Riemann approximation (see Appendix Lemma A.2 and B.7 of \cite{Mukherjee2021}) of the sum, we have that 
\begin{equation}\label{diff zeta}
\begin{split}
    \Bigg|\int_{B_{N,\alpha}} \zeta(x) dx - \frac{2}{N}\sum_{x\in R_N\cap B_{N,\alpha} } \zeta(x) \Bigg| 
    &\leq \frac{1}{2}(N^{-\frac{1}{2}+\alpha})\cdot N^{-1}\sup_{x\in B_{N,\alpha}}|\zeta'(x)|\cr
    &= \mathcal{O}(N^{-\frac{1}{2}+\alpha}\cdot N^{-1} \cdot N^{\frac{1}{2}+\alpha})\zeta(x_N^*)\cr
    &= \mathcal{O}(N^{-1+2\alpha})\zeta(x_N^*)
\end{split}
\end{equation}

and the binomial coefficient in \eqref{z m uniq} can be approximated as
\begin{equation}\label{bionomial coef muniq}
    \binom{N}{\frac{N(1+x)}{2}} = \sqrt{\frac{2}{\pi N(1-x^2)}} \; e^{-NI(x)} \; (1+\mathcal{O}(N^{-1})).
\end{equation}

It follows from equations \eqref{diff zeta} and \eqref{bionomial coef muniq}  and the Laplace approximation (see Appendix Lemma A.3 of \cite{Mukherjee2021}) of an integral over a shrinking interval $B_{N,\alpha}$ that: 
\begin{equation}\label{sum zeta}
\begin{split}
\sum_{x\in R_N\cap B_{N,\alpha} } \zeta(m) 
&= \frac{N}{2} \int_{B_{N,\alpha}} \zeta(x) dx + \mathcal{O}(N^{2\alpha})\zeta(x_N^*)\cr
&= \frac{N}{2}\int_{B_{N,\alpha}} \sqrt{\frac{2}{\pi N(1-x^2)}} e^{N\phi_N(x)} (1+\mathcal{O}(N^{-1})) dx \cr
& \quad + \mathcal{O}(N^{2\alpha}) \cdot \Bigg[ \sqrt{\frac{2}{\pi N(1-x_N^{*^2})}} e^{N\phi_N(x_N^*)}\; (1+\mathcal{O}(N^{-1}))\Bigg] \cr
&= \frac{\sqrt{N}}{2} \sqrt{\frac{2\pi}{N|\phi_N''(x_N^*)|}} \sqrt{\frac{2}{\pi(1-x_N^{*^2})}} e^{N\phi_N(x_N^*)}\; (1+\mathcal{O}(N^{-\frac{1}{2}+3\alpha})) \cr 
&= \frac{e^{N\phi_N(x_N^*)}}{\sqrt{(x_N^{*^2} -1)\phi_N''(x_N^*)}}\cdot (1+\mathcal{O}(N^{-\frac{1}{2}+3\alpha})).
\end{split}
\end{equation}
Therefore,
\begin{equation}\label{Z m* approx}
Z_N =\frac{e^{N\phi_N(x_N^*)}}{\sqrt{(x_N^{*^2} -1)\phi''(x_N^*)}}\cdot (1+\mathcal{O}(N^{-\frac{1}{2}+3\alpha})).
\end{equation}
This completes the proof of Lemma \ref{Appendix: lemma ZN exp unique}.
\end{proof}

\begin{lemma}\label{Appendix: lemma ZN exp many max}
Suppose $\phi(x)$ has $S\in\mathbb{N}$ global maximizers $x_i$  such that $\phi''(x_i)<0$. For $i\leq S$, let  $A_i\subset[-1,1]$ be an interval    such that $x_i\in int (A_i)$ is the unique maximizer of $\phi$ on $cl(A_i)$. Then for $N$ large enough $\phi_N$ has a unique global maximizer $x_{i,N}\to x_i$ on $A_i$ with $\phi''_N(x_{i,N})<0$ and for $\alpha\in\big(0,\frac{1}{6}\big]$, one has 

\begin{equation}\label{gibbs mi}
   \mu_N(m_N\in B^c_{N,\alpha,S}) = \exp\bigg\{\frac{1}{2}N^{2\alpha}\max_{i\leq S}\phi''_N(x_{i,N})\bigg\}\mathcal{O}(N^{\frac{3}{2}})
\end{equation}
where $B_{N,\alpha,S}=\bigcup_{i\leq S} B_{N,\alpha}(x_{i,N})$, moreover  the restricted partition function on $A_i$  can be expanded as,
\begin{equation}\label{Z|A mi}
    Z_N\big|_{A_i}= \dfrac{e^{N\phi_N(x_{i,N})}}{\sqrt{(x_{i,N}^{2}-1)\phi_N''(x_{i,N})}}\bigg(1+\mathcal{O}\bigg(N^{-\frac{1}{2}+\alpha}\bigg) \bigg)
\end{equation}
and the unrestricted partition function can be expanded as,
\begin{equation}\label{Z mi}
    Z_N = \sum_{i\leq S} \dfrac{e^{N\phi_N(x_{i,N})}}{\sqrt{(x_{i,N}^{2}-1)\phi_N''(x_{i,N})}}\bigg(1+\mathcal{O}\bigg(N^{-\frac{1}{2}+\alpha}\bigg) \bigg).
\end{equation}
Note that, here, $int (A_i)$ and $cl(A_i)$ denote the interior and closure of $A_i$, respectively.
\end{lemma}
\begin{proof}

 The fact that for $N$ large enough $\phi_N$ has a unique maximizer $x_{i,N}\to x_i$ with $\phi''(x_{i,N})<0$ can be proved applying to the function $\phi_N$  restricted to $cl(A_i)$ and using the same argument of Lemma \ref{Appendix: lemma ZN exp unique}. 
 
 Clearly, for $N$ large enough, $B_{N,\alpha}(x_{i,N})\subset A_i$ and 
\begin{equation}\label{conc mi}
    \mu_{N}(m_N(\sigma)\in B_{N,\alpha}^c(x_{i,N})|m_N(\sigma)\in A_i) = \exp\bigg\{\frac{1}{2}N^{2\alpha}\phi_N''(x_{i,N})\bigg\}\mathcal{O}(N^{\frac{3}{2}})
\end{equation}
following a step-by-step argument used to prove equation \eqref{conc1}.

Now, for $i \leq S$ and $N$ large enough, one has that $A_i\setminus B_{N,\alpha}(x_{i,N})= A_i\setminus B_{N,\alpha,S}$ and then 
$ \mu_N(m_N(\sigma)\in B_{N,\alpha}^c(x_{i,N})\big|m_N(\sigma)\in A_i) =  \mu_N(m_N(\sigma)\in B_{N,\alpha,S}^c|m_N(\sigma)\in A_i) $. Therefore,

\begin{equation}\label{prob_mi}
\begin{split}
    \mu_{N}(m_N(\sigma)\in B_{N,\alpha,S}^c) &= \sum_{1\leq i \leq S} \mu_{N}(m_N(\sigma)\in B_{N,\alpha,S}^c\big|m_N(\sigma)\in A_i) \mu_{N}(m_N(\sigma) \in A_i)\cr
    & \leq \exp\bigg\{\frac{1}{2}N^{2\alpha}\max_{1\leq i \leq S}\phi''(x_i)\bigg\}\mathcal{O}(N^{\frac{3}{2}}) \sum_{1\leq i \leq S}\mu_{N}(m_N(\sigma) \in A_i) \cr
    & = \exp\bigg\{\frac{1}{2}N^{2\alpha}\max_{1\leq i \leq S}\phi''(x_i)\bigg\}\mathcal{O}(N^{\frac{3}{2}}).
\end{split}
\end{equation}
This completes the proof of equation \eqref{gibbs mi} following from equation \eqref{prob_mi}.

The proof for the asymptotic expansion of the partition function when there are multiple global maximizers of $\phi$ follows exactly the same argument for the case with unique global maximizer. 
Note that for fixed $i\leq S$ and $N$ large,  $m_N(\sigma)$ concentrates around $x_i\in A_i$ as it was shown in equation \eqref{conc mi}. Hence,
\begin{equation}
    \mu_{N}(m_N(\sigma)\in B_{N,\alpha}\big|m_N(\sigma)\in A_i)= \dfrac{1}{Z_N\big|_{A_i}} \sum_{x\in R_N\cap B_{N,\alpha}} \binom{N}{\frac{N(1+x)}{2}} \exp\big\{N\big(f_N(x)\big)\big\}.
\end{equation}
Now, following the exact computation and argument in Lemma \ref{Appendix: lemma ZN exp unique}, we have that the restricted partition function for each of the global maximizers $x_i$ can be expanded as

\begin{equation}\label{z mi}
    Z_N\big|_{A_i} = \frac{e^{N\phi_N(x_{i,N})}}{\sqrt{(x_{i,N}^{2} -1)\phi''(x_{i,N})}}\cdot (1+\mathcal{O}(N^{-\frac{1}{2}+3\alpha})).
\end{equation}
Assuming that $m_N(\sigma)$ concentrates around $S$ global maximizers $x_{i,N}$ for $i\leq S$ then, equation \eqref{Z mi} follows from  \eqref{z mi}. Hence, we have 

\begin{equation}
    Z_N = \sum_{i\leq S} Z_{N}\big|_{A_i}.
\end{equation}
\end{proof}

\section*{Acknowledgement}
Two of the authors (P.C. and G.O.) thank Janos Kertész and Cecilia Vernia for insightful discussions on the topic. This work was partially supported by Gruppo Nazionale di Fisica matematica. The authors were partially  supported by the EU H2020 ICT48 project Humane AI Net contract number 952026; by the Italian Extended Partnership PE01 - FAIR Future Artificial Intelligence Research - Proposal code PE00000013 under the MUR National Recovery and Resilience Plan; by the project PRIN 2022 - Proposal code: J53D23003690006.

\end{document}